\documentclass[aps, pra, reprint, nofootinbib]{revtex4-2}
\usepackage{graphicx}
\usepackage{dcolumn}
\usepackage{amsmath,amsthm,amssymb,scrextend}
\usepackage{thmtools}
\usepackage{physics, qcircuit}
\usepackage{natbib}
\newcommand{\Cn}{\mathcal{C}_n}
\newcommand{\Pn}{\mathcal{P}_n}
\newcommand{\paulis}{\{I,X,Y,Z\}}
\newcommand{\mes}[1]{\ket{\Phi_{#1}}}
\newcommand{\choi}[2]{\ket{\Phi_{#1}^{#2}}}
\newcommand{\stab}{\text{Stab}}
\newcommand{\kpsi}{\ket{\psi}}
\newcommand{\Ftwo}{\mathbb{F}_2}
\newcommand{\Ctil}{\widetilde{C}}

\newcommand{\Pcd}{P^{\vec{c}, \vec{d}}}
\newcommand{\Pab}{P^{\vec{a}, \vec{b}}}
\newcommand{\Phf}{P^{\vec{h}, \vec{f}}}
\newcommand{\binset}[1]{\{0,1\}^{#1}}
\newcommand{\vecal}{\vec{\alpha}}

\newcommand{\sysA}{\mathcal{A}}
\newcommand{\sysB}{\mathcal{B}}
\newcommand{\eps}[1]{\epsilon^{(#1)}}
\renewcommand{\ij}{\vec{i},\vec{j}}
\newcommand{\kl}{\vec{k}, \vec{\ell}}
\newcommand{\albet}{\vecal, \vec{\beta}}
\newcommand{\gamlam}{\vec{\gamma}, \vec{\lambda}}
\newcommand{\K}{\vec{K}}
\newcommand{\J}{\vec{J}}
\newtheorem{lemma}{Lemma}

\newtheorem{theorem}{Theorem}

\begin{document}

\title{Process Tomography for Clifford Unitaries}
\date{\today}
\author{Timothy Skaras}
\email{tgs52@cornell.edu}
\author{Paul Ginsparg}
\affiliation{Department of Physics, Cornell University, Ithaca, NY 14853, USA}
\begin{abstract}
We present an algorithm for performing quantum process tomography on an unknown $n$-qubit unitary $C$ from the Clifford group. Our algorithm uses Bell basis measurements to deterministically learn $C$ with $4n + 3$ queries, which is the asymptotically optimal query complexity. In contrast to previous algorithms that required access to $C^\dagger$ to achieve optimal query complexity, our algorithm achieves the same performance without querying $C^\dagger$. Additionally, we show the algorithm is robust to perturbations and can efficiently learn the closest Clifford to an unknown non-Clifford unitary $U$ using query overhead that is logarithmic in the number of qubits. 
\end{abstract}

\maketitle

\section{Introduction}

In the general case, learning an unknown quantum state or unknown quantum channel requires resources exponential in system size: quantum state tomography requires exponential copy complexity\footnote{By copy complexity, we mean the number of copies of the unknown state required to complete the task.}  \cite{haah2016sample} and quantum process tomography requires exponential query complexity \cite{haah2023query}. It is worth noting, however, that these bounds apply when there is no prior information about the unknown state or process. Recent work \cite{zhao2024learning} has highlighted that although most unitaries or states have exponential gate complexity, such states and processes are not physical in the sense that they cannot arise in polynomial time in real systems governed by local Hamiltonians. Indeed, if the unknown state or process is known to have certain structure, in many cases the learning task can be performed much more efficiently.

For instance, various work has investigated the copy complexity of quantum state tomography on an unknown stabilizer state. From Holevo's theorem \cite{holevo1973bounds}, the copy complexity of state tomography on an $n$-qubit stabilizer state is $\Omega(n)$. Gottesman and Aaronson \cite{gottesman_aaronson_08080052} described two algorithms for this task: one which used $\mathcal{O}(n^2)$ single-copy measurements and another which used entangled measurements on $\mathcal{O}(n)$ copies. While the latter is theoretically optimal, collective measurement on this many copies is impractical. Montanaro \cite{montanaro2017learning} described an algorithm which used two-copy Bell basis measurements to perform state tomography with overall $\mathcal{O}(n)$ copy complexity.

For the task of process tomography, Low \cite{low2009learning} proposed an algorithm for learning an unknown Clifford $C$ with $\mathcal{O}(n)$ query complexity. This query complexity was shown to be optimal, but the algorithm required access to the inverse process $C^\dagger$. More recent work \cite{lai2022learning} describes a probabilistic algorithm for learning $C$ without $C^\dagger$, but the method uses $\mathcal{O}(n^2)$ queries and is therefore not asymptotically optimal. In this work, we propose a deterministic $4n + 3$ query complexity algorithm for process tomography on a Clifford unitary which does not require access to $C^\dagger$. Our algorithm queries $C$ on Bell basis states to deterministically learn how $C$ transforms Pauli operators. Lastly, we consider the case in which we have oracle access to an unknown non-Clifford unitary $U$. We find that by repeating each step of our original algorithm multiple times, we can learn the closest Clifford to $U$ with minimal overhead.

\section{Background}
The $n$-qubit Pauli group $\mathcal{P}_n$ is generated by the single-qubit Pauli operators 
\[
X_j =  
\begin{pmatrix}
    0 & 1\\
    1 & 0
\end{pmatrix}, \,
Y_j = 
\begin{pmatrix}
    0 & -i\\
    i & 0
\end{pmatrix}, \,
Z_j = 
\begin{pmatrix}
    1 & 0\\
    0 & -1
\end{pmatrix},
\]
which are understood as acting on the $j^{th}$ qubit for $j = 1,\ldots, n$.
We denote Pauli operators using pairs of bits:
\begin{equation}
P^{00} = I, \quad P^{01} = X, \quad P^{10} = Z, \quad P^{11} = Y.
\end{equation}
Using bit strings $\vec{a} , \vec{b}\in \binset{n}$ , we extend this notation to denote $n$-qubit Pauli operators
\[
\Pab = P^{a_1 b_1}\otimes \cdots \otimes P^{a_n b_n}.
\]
This notation is convenient because multiplying Pauli operators is the same as modulo 2 addition on the labels up to a global phase $\Pab \Pcd = \kappa P^{\vec{a}\oplus\vec{c}, \vec{b}\oplus\vec{d}}$ where $\kappa$ is a complex phase factor.

We can decompose any $n$-qubit Pauli as a product of single qubit $X$ and $Z$ operators with a global phase factor
\begin{equation}
    P^{\vec{a},\vec{b}} = (-i)^{\vec{a}\cdot\vec{b}}\prod_{i=1}^n Z_i^{a_i}X_i^{b_i}.
    \label{eqn:PauliDecompose}
\end{equation}
Once we include multiplicative factors of $\pm 1$ and $\pm i$, these Pauli operators form a closed group under matrix multiplication known as the Pauli group, which can be written $\Pn = \{i^k \Pab | \vec{a}, \vec{b} \in \binset{n}, k \in \{0,1,2,3\} \}$. 

Our circuits will make use of measurements in the Bell basis. On two qubits, there are four Bell pair states
\begin{align*}
\ket{\Phi^{00}} = \frac{1}{\sqrt2}(\ket{00}+\ket{11})&,\quad
\ket{\Phi^{01}} = \frac{1}{\sqrt2}(\ket{01}+\ket{10}),\\
\ket{\Phi^{10}} = \frac{1}{\sqrt2}(\ket{00}-\ket{11})&,\quad
\ket{\Phi^{11}} = \frac{1}{\sqrt2}(\ket{01}-\ket{10}).
\end{align*}
By inspection, we see that any Bell pair can be written as a Pauli acting on one half of $\ket{\Phi^{00}}$:
\begin{equation}
\ket{\Phi^{ij}} = Z_1^i X_1^j\ket{\Phi^{00}}
= X_2^j Z_2^i \ket{\Phi^{00}}.
\label{eqn:SingleBellPauli}
\end{equation}
Generalizing to larger systems, we denote a Bell state on $2n$ qubits as
\[
|\Phi^{\vec{i},\vec{j}}\rangle
= \ket{\Phi^{i_1j_1}}_{1,n+1} \cdots \ket{\Phi^{i_n j_n}}_{n,2n}
\]
where $\vec{i},\vec{j} \in \binset{n}$ and $\ket{\Phi^{i_1j_1}}_{r,s}$ denotes a Bell pair on qubits $r$ and $s$. If we have a Bell state on two registers of $n$ qubits $\mathcal{A}$ and $\mathcal{B}$, we can write
\begin{equation}
|\Phi^{\vec{i},\vec{j}}\rangle_{\sysA\sysB}
= i^{\vec{i}\cdot \vec{j}} P_\sysA^{\vec{i},\vec{j}}|\Phi^{\vec{0},\vec{0}}\rangle_{\sysA\sysB}
= (-i)^{\vec{i}\cdot \vec{j}} P_\sysB^{\vec{i},\vec{j}}|\Phi^{\vec{0},\vec{0}}\rangle_{\sysA\sysB},
\label{eqn:BellPauli}
\end{equation}
which generalizes the observation in eq.~(\ref{eqn:SingleBellPauli}).
The state $|\Phi^{\vec{0},\vec{0}}\rangle$ is a maximally entangled state which, for simplicity, we will sometimes denote as $\mes{d}$, where $d$ is the dimension of each subsystem $\mathcal{A}$ and $\mathcal{B}$.
This state can also be written as
\begin{equation}
\mes{d} = \frac{1}{\sqrt d}\sum_{i=0}^{d-1} \ket{i,i}.
\end{equation}

The $n$-qubit Clifford group $\mathcal{C}_n$ is defined as the normalizer for $\mathcal{P}_n$:
\[
\mathcal{C}_n = \{U \in U(2^n) \, | \, U \mathcal{P}_n U^\dagger \subseteq \mathcal{P}_n\}.
\]
The Gottesman-Chuang hierarchy \cite{gottesman1999demonstrating} generalizes the Clifford group, and we denote the $k^{th}$ level of this hierarchy on $n$ qubits as $\mathcal{C}^{(k)}_n$. This hierarchy is defined recursively as
\[
\mathcal{C}_n^{(k)} = \{U \in U(2^n) \, | \, U \mathcal{P}_n U^\dagger \subseteq \mathcal{C}_n^{(k-1)} \},
\]
with the assumption that $\mathcal{C}_n^{(1)} = \Pn$. Note that the Clifford group $\Cn = \Cn^{(2)}$ is the second level in this hierarchy.

Without loss of generality, let us suppose $C \in \Cn$ is an unknown Clifford unitary such that
\begin{equation}
C Z_i C^\dagger = (-1)^{f_i} P^{\vec{a}_i, \vec{b}_i} \qquad 
C X_i C^\dagger = (-1)^{h_i} P^{\vec{c}_i, \vec{d}_i},
\label{eqn:CliffordBasis}
\end{equation}
where $\vec{a}_i, \vec{b}_i, \vec{c}_i,\vec{d}_i \in \binset{n}$ and $f_i, h_i \in \binset{}$. For any element $C$ of the Clifford group, we can decompose it as
\begin{equation}
C = \Ctil \Phf
\label{eqn:DecomposeC}
\end{equation}
where 
\begin{equation}
\Ctil Z_i \Ctil^\dagger =  P^{\vec{a}_i, \vec{b}_i}, \qquad 
\Ctil X_i \Ctil^\dagger =  P^{\vec{c}_i, \vec{d}_i},
\label{eqn:CtilDefinition}
\end{equation}
and $\vec{h}, \vec{f}$ have their elements determined by $h_i, f_i$ in eq.~(\ref{eqn:CliffordBasis}). Eq.~(\ref{eqn:DecomposeC}) follows from the fact that if $C$ flips the sign of $Z_i$ under conjugation (i.e., $f_i = 1$), this is the same as $\Phf$ having an $X$ Pauli operator on the $i^{th}$ qubit.

More formally, the Clifford decomposition in eq.~(\ref{eqn:DecomposeC}) is a manifestation of the fact that $\Ctil \in \Cn / \Pn$. We will use the fact that the quotient group $\Cn / \Pn$ is isomorphic to the symplectic group over the binary field $\mathbb{F}_2$ \cite{koenig2014efficiently}:
\[
\Cn / \Pn \cong \text{Sp}(2n, \mathbb{F}_2).
\]
We define the symplectic group as being $2n \times 2n$ matrices $S$ with entries in $\mathbb{F}_2$ such that
\begin{equation}
S^T \Lambda(n) S = \Lambda(n) \equiv 
\begin{pmatrix}
    0_n & I_n\\
    I_n & 0_n 
\end{pmatrix}.
\label{eqn:SymplecticDef}
\end{equation}
This relationship is equivalent to enforcing the constraint that $\Ctil$ must map Paulis $Z_i$ and $X_i$ in such a way that preserves their algebra, i.e., preserves all pairwise commutation and anti-commutation relationships. In this representation, we map $Z_i, X_i$ to standard basis elements of the vector space $\mathbb{F}_2^{2n}$
\[
Z_i = 
\begin{pmatrix}
    \vec{e}_i\\
    \vec{0}
\end{pmatrix}, \quad
X_i = 
\begin{pmatrix}
    \vec{0}\\
    \vec{e}_i
\end{pmatrix}.
\]
An arbitrary Pauli element up to multiplicative phase factor is therefore
\[
\Pab = 
\begin{pmatrix}
    \vec{a}\\
    \vec{b}\\
\end{pmatrix},
\]
where $\vec{a}$ and $\vec{b}$ are understood to be column vectors.
 In the symplectic representation, the first $n$ columns of $S$ determine how each $Z_i$ gets mapped and the last $n$ columns determine how each $X_i$ gets mapped.

More concretely, we can write the transformation described by Clifford $\Ctil$ in eq.~(\ref{eqn:CtilDefinition}) as a symplectic matrix $S$
\[
S = 
\begin{pmatrix}
    \vec{a}_1 & \cdots & \vec{a}_n & \vec{c}_1 & \cdots & \vec{c}_n\\
\\[-8pt]
\vec{b}_1 & \cdots & \vec{b}_n & \vec{d}_1 & \cdots & \vec{d}_n\\
\end{pmatrix}
\equiv
\begin{pmatrix}
    A & C\\
    B & D
\end{pmatrix},
\]
where $A,B,C,D \in \Ftwo^{n \times n}$ are matrices defined to have columns given by $\vec{a}_i, \vec{b}_i, \vec{c}_i, \vec{d}_i$ respectively. We use the symplectic relationship (\ref{eqn:SymplecticDef}) to write an explicit expression for the inverse of $S$, which is the symplectic representation for $\Ctil^\dagger$. Multiplying both sides of eq.~(\ref{eqn:SymplecticDef}) by $\Lambda$ gives us $\Lambda S^T \Lambda S = I$, which implies
\begin{equation}
S^{-1} = \Lambda S^T \Lambda = 
\begin{pmatrix}
    D^T & C^T\\
    B^T & A^T
\end{pmatrix}.
\label{eqn:SymplecticInv}
\end{equation}

A Pauli operator $P \in \Pn$ is said to stabilize a state $\ket{\psi}$ if $P\ket{\psi} = \ket{\psi}$. 
An $n$-qubit stabilizer state is a quantum state that can be prepared by applying a Clifford unitary to the state $\ket{0}^{\otimes n}$.
We denote as $\stab(\ket{\psi}) \subseteq \Pn$ the set of Pauli operators that stabilize $\kpsi$.
This subset of the Pauli group forms a subgroup under matrix multiplication and is called the stabilizer group for $\kpsi$.

\section{Learning Clifford Unitaries}

In this section, we describe our algorithm to learn an unknown $n$-qubit Clifford unitary $C$.
The algorithm is deterministic and guaranteed to learn the unknown Clifford with only $4n + 3$ queries to $C$.
In the first stage of the algorithm, we use a subroutine based on the Twin-$C$ circuit depicted in fig.~\ref{fig:CC} to learn the symplectic representation $S$ used to describe $\Ctil$. In the second stage we use the circuit in fig.~\ref{fig:LearnPauli} to learn $\Phf$.
By eq.~(\ref{eqn:DecomposeC}), these two unitaries determine $C$.

\subsection{Algorithm}

To learn $S$ in the first stage of the algorithm, we need the following Lemma which shows that the measurement output of the Twin-$C$ circuit in fig.~\ref{fig:CC} is an affine transformation of the input bit string.

\begin{lemma}
    Let $S \in \Ftwo^{2n\times2n}$ be the symplectic representation for unknown Clifford quotient group element $\Ctil \in \Cn /\Pn$. If the two registers of the Twin-$C$ circuit are initialized with bit string $\J^T = (\vec{i}^T, \vec{j}^T)$ 
    then the output $\vec{K}^T = (\vec{k}^T, \vec{\ell}^T)$ will satisfy
    \begin{equation}
        \vec{K} = S \vec{J} + \vec{F}_0,
        \label{eqn:AffineTransform}
    \end{equation}
    where $\vec{F}_0\in \binset{2n}$ depends on $S$ but not $\vec{J}$.
    \label{lemma:AffineTransform}
\end{lemma}

\begin{proof}
We calculate the effect of this circuit
\begin{align*}
\stab(|\vec{i}, \vec{j}\rangle) 
&= \langle (-1)^{i_1}Z_1, \ldots, (-1)^{j_1}Z_{n+1},\ldots \rangle\\
\xrightarrow{H^{\otimes n}\otimes I}  \,&\langle (-1)^{i_1}X_1, \ldots, (-1)^{j_1}Z_{n+1},\ldots \rangle\\
\xrightarrow{CNOT} \,&\langle (-1)^{i_1}X_1X_{n+1}, \ldots, (-1)^{j_1}Z_1 Z_{n+1},\ldots \rangle\\
\xrightarrow{C\otimes C} \,&\langle (-1)^{i_1} P^{\vec{c}_1,\vec{d}_1}_{\mathcal{A}}P^{\vec{c}_1,\vec{d}_1}_{\mathcal{B}}, \ldots, (-1)^{j_1}P^{\vec{a}_1,\vec{b}_1}_{\mathcal{A}} P^{\vec{a}_1,\vec{b}_1}_{\mathcal{B}},\ldots \rangle\\
\xrightarrow{CNOT} \,&\langle (-1)^{i_1 + \vec{c}_1\cdot \vec{d}_1} P^{\vec{0},\vec{d}_1}_{\mathcal{A}}P^{\vec{c}_1,\vec{0}}_{\mathcal{B}}, \ldots, \\
&\qquad\qquad (-1)^{j_1 + \vec{a}_1\cdot \vec{b}_1}P^{\vec{0},\vec{b}_1}_{\mathcal{A}} P^{\vec{a}_1,\vec{0}}_{\mathcal{B}},\ldots \rangle\\
\xrightarrow{H^{\otimes n} \otimes I} \,&\langle (-1)^{i_1 + \vec{c}_1\cdot \vec{d}_1} P^{\vec{d}_1,\vec{0}}_{\mathcal{A}}P^{\vec{c}_1,\vec{0}}_{\mathcal{B}}, \ldots, \\
&\qquad\qquad (-1)^{j_1 + \vec{a}_1\cdot \vec{b}_1}P^{\vec{b}_1,\vec{0}}_{\mathcal{A}} P^{\vec{a}_1,\vec{0}}_{\mathcal{B}},\ldots \rangle\\
&= \stab(\ket{\psi_f}).
\end{align*}
Because all generators of the stabilizer group are composed of $Z_i$ operators, the final state is a $Z$-basis eigenstate, and the measurement output is deterministic.
Let $\vec{k}, \vec{\ell} \in \binset{n}$ be the bit strings measured on registers $A$ and $B$.
For the $r^{th}$ qubit, we know either $Z_r$ or $-Z_r$ is in $\stab(\ket{\psi_f})$.
The sign in front of $Z_r$ determines whether $0$ or $1$ is measured on the $r^{th}$ qubit.
To calculate this sign, we write the generators for $\stab(\ket{\psi_f})$ as a matrix and keep track of the signs in the rightmost column while performing Gaussian elimination.
Writing the generators for $\stab(\ket{\psi_f})$ in tableau form gives
\begin{equation}
\left[
\begin{array}{c|c}
\ \vec{d}_1^T \qquad \vec{c}_1^T & i_1 + \vec{c}_1\cdot\vec{d}_1 \\
\vdots & \vdots \\
\ \vec{d}_n^T \qquad \vec{c}_n^T & i_n + \vec{c}_n\cdot\vec{d}_n \\
\vec{b}_1^T \qquad  \vec{a}_1^T  & j_1 + \vec{a}_1\cdot\vec{b}_1 \\
\vdots & \vdots \\
\ \vec{b}_n^T \qquad  \vec{a}_n^T  & j_n + \vec{a}_n\cdot\vec{b}_n 
\end{array}
\right] \equiv
\left[    
\begin{array}{c|c}
M & \vec{J} + \vec{F}
\end{array}
\right],
\label{eqn:Tableau}
\end{equation}
where 
\[M = 
\begin{pmatrix}
    D^T & C^T\\
    B^T & A^T\\
\end{pmatrix},
\quad
\vec{J} =
\begin{pmatrix}
     \vec{i}  \\
     \vec{j} 
\end{pmatrix},
\quad
\vec{F} = 
\begin{pmatrix}
    \vec{c}_1\cdot\vec{d}_1\\
    \vdots\\
    \vec{c}_n\cdot\vec{d}_n\\
    \vec{a}_1\cdot\vec{b}_1\\
    \vdots \\
    \vec{a}_n\cdot\vec{b}_n\\
\end{pmatrix}.
\]
We calculate the measurement outputs $\vec{k}, \vec{\ell}$ by inverting $M$ and applying it to the tableau in eq.~(\ref{eqn:Tableau})
\begin{equation*}
\begin{pmatrix}
\vec{k}  \\
 \vec{\ell}
\end{pmatrix}
=\K
= M^{-1}\vec{J} + M^{-1}\vec{F},
\end{equation*}
which implies the output measurement is an affine transformation of the input $\vec{J}$. To compute $M^{-1}$, observe that $M$ equals the formula for $S^{-1}$ in eq.~(\ref{eqn:SymplecticInv}). Hence, $M^{-1} = S$, where $S$ is the symplectic representation for the unknown $\Ctil$. Thus,
\begin{equation}
\K
= S
\J
+ \vec{F}_0,
\end{equation}
where $\vec{F}_0 = S\vec{F}$.

\end{proof}

\begin{figure}[t]
\centering
\[
\hspace{2.5em}
\Qcircuit @C=1.1em @R=1.5em  {
    \lstick{|\vec{i}\rangle_{\mathcal{A}}} & \gate{\resizebox{2.5em}{!}{$H^{\otimes n}$}} & \ctrl{1} &  \gate{\resizebox{1.25em}{!}{$C$}}  & \ctrl{1} & \gate{\resizebox{2.5em}{!}{$H^{\otimes n}$}} & \meter&\cw & \rstick{\hspace{-1em}\vec{k}} \\
    \lstick{|\vec{j}\rangle_{\mathcal{B}}} & \qw & \targ & \gate{\resizebox{1.25em}{!}{$C$}} & \targ & \qw & \meter&\cw &\rstick{\hspace{-1em}\vec{\ell}}  
}
\]
\caption{[Twin-$C$ Circuit] The circuit is initialized with bit strings $\vec{i}, \vec{j} \in \binset{n}$. We query $C$ on both registers and perform a Bell basis measurement to obtain outcomes $\vec{k}$ and $\vec{\ell}$ on registers $\mathcal{A}$ and $\mathcal{B}$. In later sections, we use the same circuit but query some more general unknown unitary $U$ instead of the Clifford $C$. In such cases, we denote the circuit a Twin-$U$ circuit.}
\label{fig:CC}
\end{figure}
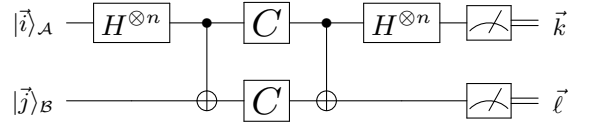

Lemma \ref{lemma:AffineTransform} shows that we can learn $\Ctil$ by executing the circuit of fig.~\ref{fig:CC} on basis bit strings of $\Ftwo^{2n}$.
To complete the learning algorithm, we recall that we can learn $P^{\vec{h},\vec{f}}$ with a single Bell basis measurement, as shown in Proposition 21 of ref.~\cite{montanaro2008quantum}.
We provide a proof to clarify the later stages of the learning task. 

\begin{lemma}[ref. \cite{montanaro2008quantum}]
An unknown $n$-qubit Pauli $P \in \paulis^{\otimes n}$ can be identified using a single query and a Bell basis measurement.
\label{lemma:DistinguishP}
\end{lemma}

\begin{proof}
First, create a maximally entangled state $\mes{2^n}$ on two registers of $n$ qubits.
Query $P$ on the second register to create the Choi state \cite{choi1975completely,jamiolkowski1972linear, khatri2020principles} 
$$\choi{2^n}{P} =  I\otimes P\mes{2^n}.$$
Without loss of generality, let $P = \Pab$ for some $\vec{a},\vec{b}\in \binset{n}$.
From eq.~(\ref{eqn:PauliDecompose}), we know $\Pab = (-i)^{\vec{a}\cdot\vec{b}}\prod_{i=1}^n Z_i^{a_i}X_i^{b_i}$. Then, measuring the Choi state in the Bell basis gives
\begin{align*}
    \stab(\mes{2^n}) = 
    &\,\langle X_1 X_{n+1}, \ldots, Z_1 Z_{n+1}, \ldots \rangle\\
    \xrightarrow{I\otimes P} &\,\langle (-1)^{a_1} X_1 X_{n+1}, \ldots, (-1)^{b_1}Z_1 Z_{n+1}, \ldots\rangle\\
    \xrightarrow{CNOT} &\,\langle (-1)^{a_1} X_1 , \ldots, (-1)^{b_1}Z_{n+1}, \ldots\rangle\\
    \xrightarrow{H^{\otimes n}\otimes I} &\,\langle (-1)^{a_1} Z_1 , \ldots, (-1)^{b_1}Z_{n+1}, \ldots\rangle\\
    = &\,\, \stab(|\vec{a}, \vec{b}\rangle).
\end{align*}
In the third and fourth lines, we have applied gates to transform into the Bell basis, and we see the final state is the $Z$-eigenstate $|\vec{a}, \vec{b}\rangle$.
Thus, we can determine $P$ up to global phase with a single $Z$-basis measurement by reading off the bit strings in the first and second registers.

\end{proof}

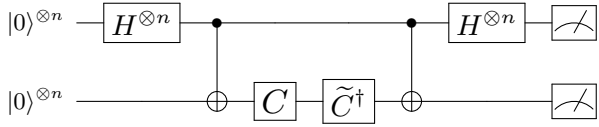
\begin{figure}[t]
\centering
\[
\hspace{2.5em}
\Qcircuit @C=1.1em @R=1.5em  {
    \lstick{|0\rangle^{\otimes n}} & \gate{\resizebox{2.5em}{!}{$H^{\otimes n}$}} & \ctrl{1} &  \qw & \qw  & \ctrl{1} & \gate{\resizebox{2.5em}{!}{$H^{\otimes n}$}} & \meter \\
    \lstick{|0\rangle^{\otimes n}} & \qw & \targ & \gate{\resizebox{1.1em}{!}{$C$}} & \gate{\resizebox{1.5em}{!}{$\Ctil^\dagger$}} & \targ & \qw & \meter
}
\]
\caption{Create the maximally entangled state, query $C$ on the second register, apply $\Ctil^\dagger$ after learning $\Ctil$, and perform a Bell basis measurement.}
\label{fig:LearnPauli}
\end{figure}

We now prove our main result regarding the process tomography algorithm:

\begin{theorem}
Given oracle access to an unknown Clifford unitary $C \in \Cn$, we can determine $C$ up to global phase with  $4n + 3$ queries to $C$ and $\mathcal{O}(\frac{n^2}{\log n})$ gates.
\end{theorem}

\begin{proof}
Recall from eq.~(\ref{eqn:DecomposeC}), that up to global phase we can write $C = \Ctil \Phf$ for some $\Ctil \in \Cn /\Pn$ and $n$-qubit Pauli $\Phf$.
Without loss of generality, assume the action of $\Ctil$ on Pauli operators is given by eq.~(\ref{eqn:CtilDefinition}).

In the first stage of the algorithm, we learn $\Ctil$ using the Twin-$C$ circuit of fig.~\ref{fig:CC}.
We execute the circuit on input $\vec{J}_0^T = (\vec{0 }^T, \vec{0}^T)$ to obtain output $\vec{K}_0$.
By Lemma \ref{lemma:AffineTransform}, this will output $\vec{K}_0 = \vec{F}_0$.
We run the circuit $2n$ more times on inputs $\vec{J}_1, \ldots, \vec{J}_{2 n}$ where $\vec{J}_i =\vec{e}_i $.
Let $\vec{K}_1, \ldots, \vec{K}_{2n}$ be the measurement outputs from these circuit runs.
Add $\vec{K}_0$ to each output to obtain a new set of bit strings $\vec{K}^\prime_1,\ldots,\vec{K}_{2n}^\prime$.
By Lemma \ref{lemma:AffineTransform}, we conclude that $\vec{K}_i^\prime$ is the $i^{th}$ column of $S$
\[
\vec{K}^\prime_i \equiv \vec{K}_i +\vec{K}_0 = S \vec{J}_i,
\]
thereby determining $\Ctil$.
The $2n +1$ runs of circuit \ref{fig:CC} employ $4n + 2$ queries of $C$.

In the second stage of the algorithm, we use Lemma \ref{lemma:DistinguishP} and the circuit of fig.~\ref{fig:LearnPauli} to learn $\Phf$.
Since we have learned $\Ctil$, we can compute its inverse using eq.~(\ref{eqn:SymplecticInv}) and compile the inverse circuit using $\mathcal{O}(\frac{n^2}{\log n})$ gates \cite{aaronson2004improved}.
Querying $C$ and applying $\Ctil^\dagger$ is equivalent to querying the unknown Pauli $\Phf$:
\[
\Ctil^\dagger C = \Ctil^\dagger (\Ctil \Phf) = \Phf.
\]
From Lemma \ref{lemma:DistinguishP}, the circuit in fig.~\ref{fig:LearnPauli} outputs $\vec{h}$ on the first register and $\vec{f}$ on the second register, determining $\Phf$. 
Together with $\Ctil$, this completes the identification of $C$.
The circuit in fig.~\ref{fig:LearnPauli} queries $C$ one additional time, so the total query complexity for the algorithm is $4n + 3$.

\end{proof}

It has been proven \cite{low2009learning} that learning an unknown $n$-qubit Clifford unitary requires at least $n$ queries to the unknown Clifford. Our algorithm, which uses $\mathcal{O}(n)$ queries, therefore has asymptotically optimal query complexity.

\subsection{Discussion}

Low's algorithm \cite{low2009learning} used $C$ and $C^\dagger$ to query unknown Paulis $C Z_i C^\dagger$ and $C X_i C^\dagger$ on one half of a maximally entangled state, which could then be learned with a single measurement using Lemma \ref{lemma:DistinguishP}.
Our algorithm queries $C$ on both halves of $|\Phi^{\ij}\rangle$ which, with the ricochet property of maximally entangled states, produces the effect of querying $C P^{\ij}C^T$ on one half of the maximally entangled state.
This still allows us to complete the learning task because the transpose is the same as the inverse up to the multiplication by a Pauli:

\begin{lemma}
Let $C \in \Cn$ be an $n$-qubit Clifford unitary that satisfies eq.~(\ref{eqn:CliffordBasis}). Then, 
\begin{equation}
C^* = e^{i\theta} C P^{\vecal, \vec{\beta}},
\label{eqn:StarTransform}    
\end{equation}
where $e^{i\theta}$ is a phase factor, $\alpha_i = \vec{c}_i \cdot \vec{d}_i$, $\beta_i = \vec{a}_i \cdot \vec{b}_i$, and the complex conjugate is in the computational basis.
\end{lemma}

\begin{proof}
Taking the complex conjugate in the computational basis of the equations in (\ref{eqn:CliffordBasis}), we obtain
\[
C^* Z_i^* C^{*\dagger} = (-1)^{f_i} (P^{\vec{a}_i, \vec{b}_i})^*, \, 
C^* X_i C^{*\dagger} = (-1)^{h_i} (P^{\vec{c}_i, \vec{d}_i})^*.
\]
From the definition in eq.~(\ref{eqn:PauliDecompose}), we know $(P^{\vec{a}_i, \vec{b}_i})^* = (-1)^{\vec{a}_i \cdot \vec{b}_i} P^{\vec{a}_i, \vec{b}_i}$.
So the above equations can be simplified to
\begin{align*}
C^* Z_i C^{*\dagger} = (-1)^{f_i + \vec{a}_i \cdot \vec{b}_i} P^{\vec{a}_i, \vec{b}_i},\\ 
C^* X_i C^{*\dagger} = (-1)^{h_i + \vec{c}_i\cdot \vec{d}_i} P^{\vec{c}_i, \vec{d}_i}.
\end{align*}
Let $\vecal, \vec{\beta} \in \binset{n}$ be bit strings where the $i^{th}$ element is 
\[
\alpha_i = \vec{c}_i \cdot \vec{d}_i, \quad
\beta_i = \vec{a}_i \cdot \vec{b}_i.
\]
Let us define the unitary $U = C P^{\albet}$.
Conjugating the Pauli $Z_i$ with this unitary gives
\begin{align*}
U Z_i U^\dagger 
& = C P^{\albet} Z_i P^{\albet}C^\dagger\\
&= (-1)^{\beta_i} C Z_i C^\dagger\\
&= (-1)^{\vec{a}_i \cdot \vec{b}_i + f_i} P^{\vec{a}_i, \vec{b}_i}\\
&= C^* Z_i C^{*\dagger},
\end{align*}
where the second line follows from the fact that $P^{\albet}$ anticommutes with $Z_i$ if and only if $\beta_i = \vec{a}_i \cdot \vec{b}_i = 1$.
By analogous reasoning, it is also true that $U X_i U^\dagger = C^* X_i C^{*\dagger}$.
It follows that $U$ and $C^*$ produce identical transformations on any density matrix $\rho$.
These two unitaries, therefore, can differ by at most a global phase,
and we conclude $C^* = e^{i\theta} U$.

\end{proof}

To demonstrate more clearly the use of $C^T$ versus $C^\dagger$, we prove a more general result to be used in the following section when the unknown unitary is not in the Clifford group.
Suppose the circuit in fig.~\ref{fig:CC} queries an unknown unitary $U$.
When $U \notin \Cn$, the output will not be deterministic.
We refer to this variant of the circuit as a Twin-$U$ circuit.
Let $\vec{J}$ be the input bit string to this circuit and let $\vec{K}$ be the output, which we now treat as a random variable.
Let $P(\K = \kl|\J= \ij)$ denote the probability of measuring $\kl$ when initializing the circuit with bit strings $\ij$\footnote{Formally, $\J,\K \in \binset{2n}$ are column vectors, but for convenience we will write $\J = \ij$, where it is understood we mean a column vector formed by concatenating $\vec{i}$ and $\vec{j}$.}.

\begin{lemma}
Let $U \in U(2^n)$ be an $n$-qubit unitary. The outcome distribution $P(\K = \kl|\J= \ij)$ for the Twin-$U$ circuit will be
\begin{equation}
    P(\K = \kl|\J= \ij) = \frac{1}{d^2}\left|\Tr[U P^{\ij}U^T P^{\kl}]\right|^2.
    \label{eqn:CircuitOutputDist}
\end{equation}
\label{lemma:CircuitOutputDist}
\end{lemma}

\begin{proof}
We use eq.~(\ref{eqn:BellPauli}) and the ricochet property to calculate
\begin{align*}
P(\K=&\kl |\J= \ij) = |\langle\Phi^{\kl}| U\otimes U |\Phi^{\ij}\rangle|^2\\
&= |\langle\Phi_{d}| (I\otimes P^{\kl})U\otimes U (I\otimes P^{\ij})|\Phi_{d}\rangle|^2\\
&= |\langle\Phi_{d}| (I\otimes P^{\kl})I\otimes U (I\otimes P^{\ij})(I\otimes U^T)|\Phi_{d}\rangle|^2\\
& = \frac{1}{d^2}\left|\Tr[ P^{\kl}U P^{\ij}U^T]\right|^2,
\end{align*}
where the last line follows from the fact $\Tr[X] = d \langle\Phi_{d}|I\otimes X\mes{d}$.

\end{proof}

We now demonstrate how the Pauli term $P^{\albet}$ in eq.~(\ref{eqn:StarTransform}) leads to the constant term $\vec{F}_0$ in the affine transformation of eq.~(\ref{eqn:AffineTransform}). Consider the case where $U = C$ for some $C \in \Cn$ and apply (\ref{eqn:StarTransform}) to the expression (\ref{eqn:CircuitOutputDist})
\begin{align}
P(\K = \kl |\J =  \ij) 
&= \frac{1}{d^2}\left|\Tr[ P^{\kl} C P^{\ij} C^{* \dagger}]\right|^2\nonumber\\
&= \frac{1}{d^2}\left|\Tr[ P^{\kl} C P^{\ij}(C P^{\albet})^\dagger]\right|^2\nonumber\\
& = \frac{1}{d^2}\left|\Tr[ P^{\kl} C P^{\ij}C^\dagger C P^{\albet}C^\dagger]\right|^2\nonumber.
\end{align}
Without loss of generality, suppose $C P^{\albet}C^\dagger = (-1)^m P^{\gamlam}$ for some integer $m$ and bit strings $\gamlam \in \binset{n}$. We conclude
\begin{equation}
P(\K = \kl |\J= \ij) 
= \frac{1}{d^2}\left|\Tr[ P^{\vec{k} + \vec{\gamma}, \vec{\ell}+\vec{\lambda}}C P^{\ij}C^\dagger]\right|^2,
\end{equation}
where $\gamlam$ are fixed bit strings that depend only on $C $ and not $\ij,\kl$.
If we execute the circuit in fig.~\ref{fig:CC} with $\vec{i}=\vec{j}=\vec{0}$, the output will be $\gamlam$:
\[
P(\K =\kl |\J= \vec{0},\vec{0}) =  \frac{1}{d^2}\left|\Tr[ P^{\vec{k} + \vec{\gamma}, \vec{\ell}+\vec{\lambda}}]\right|^2 = \delta_{\vec{k},\vec{\gamma}}\delta_{\vec{\ell},\vec{\lambda}}.
\]
This is equivalent to learning $\vec{F}_0$ in eq.~(\ref{eqn:AffineTransform}). 
To learn how $\Ctil$ transforms $Z_r$, we initialize the circuit with $\vec{i} = \vec{e}_r, \vec{j}= \vec{0}$ and will measure $\vec{a}_r + \vec{\gamma}, \vec{b}_r + \vec{\lambda}$ because
\[
P(\K=\kl |\J= \vec{e}_r,\vec{0}) = \delta_{\vec{k}, \vec{a}_r + \vec{\gamma}}
\delta_{\vec{\ell}, \vec{b}_r + \vec{\lambda}}.
\]
Our algorithm simply adds the bit strings $\gamlam$ we measured with all qubits initialized to zero to learn $\vec{a}_r, \vec{b}_r$.
Having demonstrated the connection with Low's algorithm, in the next section we use eq.~(\ref{eqn:CircuitOutputDist}) to calculate the query complexity of learning the Clifford element closest to $U \notin\Cn$.

\section{Learning The Closest Clifford To An Unknown Unitary}

To extend the algorithm to the case with $U \notin\Cn$, we first need to quantify a measure of closeness of unitaries. Following \cite{low2009learning}, we use the quantity
\begin{equation}
    D(U_1, U_2) = \frac{1}{\sqrt{2d^2}} \norm{U_1\otimes U_1^* - U_2 \otimes U_2^*}_2,
\end{equation}
as a phase-insensitive measure of ``distance" between two $d\times d$ unitaries,
where $\norm{A}_2 = \sqrt{\Tr[A^\dagger A]}$ is the Frobenius norm, and $d=2^n$ in the applications below.
Expanding the norm gives the alternative expression
\begin{equation}
D(U_1, U_2) = \sqrt{1 - \frac{1}{d^2}\left|\Tr[U_1 U_2^\dagger]\right|^2},
\label{eqn:DUU}
\end{equation}
closely reflecting the distribution of $P(\K = \kl |\J= \ij)$ in eq.~(\ref{eqn:CircuitOutputDist}).
Since $D(U_1, U_2) = 0$ implies only that $U_1$ and $U_2$ are the same up to an unobservable global phase (rather than implying $U_1 = U_2$),
$D$ is not a true distance.
$D$ nonetheless obeys the triangle inequality and is unitarily invariant.

We next turn to an algorithm for learning the Clifford element $C$ closest to an unknown unitary $U$ (i.e., the unitary $C \in \Cn$ which minimizes $D(U,C)$).
For this problem to be well-defined, the closest $C$ must be unique. 
Lemma 12 from \cite{low2009learning} gives a sufficient condition:
\begin{lemma}[ref. \cite{low2009learning}]
    Let $U \in U(2^n)$ and $C \in \Cn^{(k)}$ satisfy $D(U,C) < \frac{\sqrt2}{2^k} \equiv \eps{k} $. Then $C \in \Cn^{(k)}$ is the unique closest element to $U$ up to global phase.   
    \label{lemma:LowUniqueness}
\end{lemma}


In practice, it can be difficult to implement a unitary on real hardware that is exactly in the Clifford group, due to the effects of noise.
To find the closest Clifford element to an unknown unitary, we proceed in two stages: first, execute the Twin-$U$ circuit to learn the symplectic representation $S$, and hence $\Ctil$, associated to the closest Clifford element $C$.
Second, compile the inverse $\Ctil^\dagger$ to query the unknown Pauli $\Phf$.
The circuit output will no longer be deterministic for these two stages,
but we can learn the deterministic output of the closest Clifford by  executing the Twin-$U$ circuit multiple times and taking the majority vote. 

We first establish that this majority vote scheme works to learn the closest Pauli to an unknown unitary $U$.
The result is originally proven in Proposition 22 of ref.~\cite{montanaro2008quantum}, but our proof applies for the different definition of Pauli closeness employed here.

\begin{lemma}[ref.~\cite{montanaro2008quantum}]
Given oracle access to an unknown unitary $U \in U(2^n)$, suppose $P \in \Pn$ satisfies $ D(U, P) \leq \epsilon < \eps{1}$.
Then $P$ can be determined (up to global phase) with success probability at least $1-\delta$ using 
$\mathcal{O}\left(\smash{\frac{\log(1/\delta)}{(\eps{1} - \epsilon)^2}}\right)$
queries to $U$.
\label{lemma:LearnCloseP}
\end{lemma}

\begin{proof}
Let $P \in \Pn$ be closest Pauli element to $U$ such that $D(U, P) \leq \epsilon < \eps{1}$.
By Lemma \ref{lemma:LowUniqueness}, the closest Pauli $P$ is unique.
Similar to the circuit used in Lemma \ref{lemma:DistinguishP}, we query $U$ on one half of a maximally entangled state and then perform a Bell basis measurement.
Because $U$ is in general not a Pauli operator, the outcome is not deterministic.
The probability of measuring outcome $\K = \kl$ is
\begin{align*}
P(\K=\kl) &= |\langle \Phi^{\kl}| I\otimes U \mes{2^n}|^2\\
&= |\langle \Phi_{2^n}|(I\otimes P^{\kl}) (I\otimes U) \mes{2^n}|^2\\
&= \frac{1}{d^2}\left|\Tr[P^{\kl} U]\right|^2.
\end{align*}

Let $P = \Pab$ for $\vec{a},\vec{b}\in\binset{n}$.
The probability that the measurement correctly identifies $P$ is
\begin{align}
P(\K=\vec{a},\vec{b}) 
&= \frac{1}{d^2}\left|\Tr[UP^{\vec{a},\vec{b}} ]\right|^2
\geq 1 - \epsilon^2 \,
\end{align}
where the last inequality follows from eq.~(\ref{eqn:DUU}).
For $N$ runs of the circuit with outcomes $\K_1, \ldots, \K_N$, the algorithm succeeds if the majority of these measurements equal $\vec{a},\vec{b}$.
Treating each measurement as a Bernoulli random variable with probability of success $p \geq 1 - \epsilon^2$, by Hoeffding's inequality \cite{hoeffding1994probability} the probability of the algorithm failing is
\[
P(\text{fail}) \leq 2 e^{-2N(p-1/2)^2} = \delta.
\]
Solving for $N$, we find that $\mathcal{O}\left({\frac{\log(1/\delta)}{(\eps{1}- \epsilon)^2}}\right)$
measurements ensure the algorithm succeeds with probability at least $1-\delta$.
Each measurement requires a single query of $U$, so that also provides the query complexity of the algorithm to identify $P$.
\end{proof}

To learn the closest Clifford to an unknown unitary, we apply a similar majority vote scheme to a Twin-$U$ circuit to learn how that Clifford transforms each basis Pauli operator.
For this to work, it is necessary that the unitary $UPU^T$ is sufficiently close to the Pauli $CPC^T$ generated by an actual Clifford.
The following Lemma is a version of Lemma 16 from ref.~\cite{low2009learning}, modified for this context:

\begin{lemma}
If $D(U, V) \leq \epsilon$, then for all $P \in \Pn$
\begin{equation}
    D(U P U^T,V P V^T ) \leq 2\epsilon.
\end{equation}
\label{lemma:UPU}
\end{lemma}
\begin{proof}
Let $P = \Pab$, and define unitaries $U_{ab} = U\Pab U^T$ and $V_{ab} = V\Pab V^T$, satisfying
\[
U_{ab}  = W V_{ab}W^T,
\]
where $W = UV^\dagger$.
Then, using unitary invariance
\begin{align*}
D(U_{ab}, V_{ab}) 
&= D(WV_{ab}W^T, V_{ab}) = D(WV_{ab}, V_{ab}W^*)\\
&\leq D(WV_{ab}, V_{ab}) + D(V_{ab}, V_{ab}W^*)\\
&\leq D(W, I) + D(I,W^*) = 2D(I,W)\\
&= 2D(U,V) \leq 2\epsilon.
\end{align*}

\end{proof}

With this result, we can efficiently learn the closest Clifford $C$ to an unknown unitary $U$ if $U$ is sufficiently close to $C$:

\begin{theorem}
Given oracle access to an unknown unitary $U \in U(2^n)$, suppose $C \in \Cn$ satisfies $D(U, C) \leq \epsilon < \eps{2}$.
With success probability $1- \delta$, we can learn $C$ up to global phase using $\mathcal{O}\left(\smash{\frac{n\log(n/\delta)}{(\eps{2} - \epsilon)^2}}\right)$ queries to $U$.
\end{theorem}

\begin{proof}
Let $C \in \Cn$ be a Clifford that satisfies $D(U, C) \leq \epsilon < \eps{2}$.
By Lemma \ref{lemma:LowUniqueness}, this Clifford is unique up to global phase.
Recall that eq.~(\ref{eqn:CliffordBasis}) characterizes $C$ on a basis of Pauli operators, and that $C$ can be decomposed as in eq.~(\ref{eqn:DecomposeC}).

First stage: To learn the symplectic representation $S$ associated with $\Ctil$, we execute a Twin-$U$ circuit with inputs $\J_0, \J_1,\ldots,\J_{2n}$, where $\J_0 = \vec{0}$ and $\J_i = \vec{e}_i$ for all $i=1,\ldots,2n$.
With $U$ no longer Clifford, the output is no longer guaranteed to be $S\J + \vec{F}_0$.
Nonetheless, we can show that the circuit will output $S\J + \vec{F}_0$ with probability greater than $1/2$ as long as $\epsilon < \eps{2}$.
Because the circuit succeeds most of the time, we can take the majority vote of multiple runs.

If $\K$ is the random outcome of the Twin-$U$ circuit, we calculate using Lemma \ref{lemma:CircuitOutputDist}
\begin{align*}
P(\K = S\J +\vec{F}_0 | \J = \ij) 
&= \frac{1}{d^2}\left|\Tr[U P^{\ij}U^T P^{S\J +\vec{F}_0}]\right|^2\\
&= 1- D(U P^{\ij}U^T,  P^{S\J +\vec{F}_0})^2\\
&= 1- D(U P^{\ij}U^T,  C P^{\ij}C^T)^2.
\end{align*}
Now apply Lemma \ref{lemma:UPU} to obtain the lower bound
\begin{equation}
P(\K = S\J +\vec{F}_0 | \J = \ij) 
\geq 1- 4\epsilon^2 > \frac{1}{2}.
\end{equation}
For a fixed input $\J$, this is fundamentally the same learning task solved in Lemma \ref{lemma:LearnCloseP} where we were learning the closest Pauli $P$ to an unknown unitary $U$ but now with $D(U,P) \leq 2\epsilon < \frac{1}{\sqrt2}$. 
Thus with success probability $1-\Delta$, we can learn the bit string $\vec{Q}_i \equiv S\J_i +\vec{F}_0$ that the circuit would output by querying $C$ instead of $U$ using $\mathcal{O} ( \log(1/\Delta)/ ( \eps{2} - \epsilon)^2 )$ queries to $U$.
Cycling through all inputs $\J_0,\ldots,\J_{2n}$, we learn $S$, completing the first stage of the algorithm.

Second stage: To learn $\Phf$ associated with $C$, we query $U$ and apply a compiled circuit for $\Ctil^\dagger$.
Noting from unitary invariance that
\begin{align*}
D(\Ctil^\dagger U, \Phf) &= D(\Phf C^\dagger U, \Phf) \\
&=D(C^\dagger U, I) = D(U,C) \leq \epsilon,
\end{align*}
we see that $\Ctil^\dagger U$ is close enough to $\Phf$ to apply Lemma \ref{lemma:LearnCloseP} and learn $\Phf$ with $\mathcal{O}(\log(1/\Delta)/(\eps{2} - \epsilon)^2)$ queries to $U$.

The algorithm consists of $2n + 2$ learning tasks, all of which must succeed for the algorithm overall to succeed.
Overall success probability of at least $1-\delta$ can be ensured if the failure probability $\Delta$ of any one learning task is $\Delta = \frac{\delta}{2n+2}$.
We conclude that $2n+2$ learning tasks each of which has $\mathcal{O}(\log(n/\delta)/(\eps{2} - \epsilon)^2)$ query complexity gives the algorithm a total query complexity of $\mathcal{O}\left(\smash{n\frac{\log(n/\delta)}{(\eps{2} - \epsilon)^2}}\right)$. 

\end{proof}

\section{Conclusion}

We have shown that asymptotically optimal process tomography for an unknown Clifford unitary can be accomplished without access to $C^\dagger$. Due to the deterministic nature of the algorithm,  with minimal modification the approach also permits learning the closest Clifford to an unknown unitary $U$ by using majority vote to learn how each basis Pauli is transformed.

An open question that remains is whether the algorithm can be generalized to identify unknown unitaries from higher levels in the Gottesman-Chuang hierarchy, rather than just Clifford unitaries. Our approach relied on the fact that complex conjugation of a Clifford has the same effect as multiplying it by a Pauli operator. Developing an efficient algorithm for learning unitaries above the second level of the Gottesman-Chuang hierarchy will likely require generalizing this relationship to higher level unitaries. 

\section{Acknowledgments}

This work was partially supported by US DOE grant DE-SC0020360.
We thank Erich Mueller and Mark Wilde for helpful discussions.

\bibliography{bibliography}

\end{document}